%% file: main.tex
\documentclass[twoside,leqno,twocolumn]{article}

\usepackage[letterpaper]{geometry}

\usepackage{siamproceedings}

\usepackage[T1]{fontenc}
\usepackage{amsfonts}
\usepackage{xcolor}
\usepackage{graphicx}
\usepackage{epstopdf}
\usepackage{enumitem}
\usepackage{algorithmic}
\usepackage{algorithm}
\ifpdf
  \DeclareGraphicsExtensions{.eps,.pdf,.png,.jpg}
\else
  \DeclareGraphicsExtensions{.eps}
\fi

\newsiamremark{remark}{Remark}
\newsiamremark{hypothesis}{Hypothesis}
\crefname{hypothesis}{Hypothesis}{Hypotheses}
\newsiamthm{claim}{Claim}

\usepackage{amsopn}

\begin{document}

\newcommand\relatedversion{}
\renewcommand\relatedversion{\thanks{The full version of the paper can be accessed at \protect\url{https://arxiv.org/abs/0000.00000}}} 

\title{\Large Cycle Basis Algorithms for Reducing Maximum Edge Participation}
    \author{Fan Wang\thanks{Department of Computer Science, University of California, Irvine (\email{wangf15@uci.edu}, \email{irani@ics.uci.edu})}
    \and Sandy Irani\footnotemark[1]}

\date{}

\maketitle






\input{abstract}
\input{introduction}

\input{motivation}
\input{experiments}
\input{quantum}

\input{theory}
\clearpage
\bibliographystyle{unsrt}
\bibliography{ref}

\clearpage
\input{appendix}

\end{document}

%% file: abstract.tex
\begin{abstract}
We study the problem of constructing cycle bases of graphs with low maximum edge participation, defined as the maximum number of basis cycles that share any single edge. This quantity, though less studied than total weight or length, plays a critical role in quantum fault tolerance, as it directly impacts the overhead of lattice surgery procedures used to implement an almost universal quantum gate set. Building on a recursive algorithm by Freedman and Hastings, we introduce a family of load-aware heuristics that adaptively select vertices and edges to minimize edge participation throughout the cycle basis construction. Our approach improves empirical performance on random regular graphs and on graphs derived from small quantum codes. We further analyze a simplified balls-into-bins process to establish lower bounds on edge participation. While the model differs from the cycle basis algorithm on real graphs, it captures what can be proven for our heuristics without using more complex graph theoretic properties related to the distribution of cycles in the graph. Our analysis suggests that the maximum load of all of our heuristics will be $\Omega(\log^2 n)$. Our results indicate that careful cycle basis construction can yield significant practical benefits in the design of fault-tolerant quantum systems. This question also carries theoretical interest, as it is essentially identical to the basis number of a graph—the minimum possible maximum edge participation over all cycle bases.
\end{abstract}

%% file: introduction.tex
\section{Introduction}\label{sec:intro}

A \emph{cycle basis} of an undirected graph \( G = (V, E) \) is a minimal set of simple cycles such that every cycle in \( G \) can be expressed as the symmetric difference (i.e., edgewise XOR) of some subset of these basis cycles. For any connected graph with \( n \) vertices and \( m \) edges, the dimension of the cycle space—and hence the number of basis cycles—is exactly \( m - n + 1 \)~\cite{gross2018graph}. 

One intuitive way to see this is via a spanning tree of \( G \): such a tree has exactly \( n - 1 \) edges and no cycles. Each of the remaining \( m - (n - 1) = m - n + 1 \) non-tree edges, when added back to the spanning tree, closes a unique simple cycle. These cycles, formed by combining each non-tree edge with the unique tree path connecting its endpoints, constitute a valid cycle basis. A cycle basis constructed this way is called a \emph{fundamental cycle basis}.

Much of the literature on cycle bases focuses on finding a minimum-weight cycle basis, where the weight of a cycle is the sum of its edge weights, and the weight of the basis is the sum of the weights of all basis cycles~\cite{horton1987polynomial,kavitha2008algorithm,kavitha2009cycle}. In unweighted graphs, this reduces to minimizing the total length of the basis cycles. In contrast, we consider a different structural property that disregards edge weights: the \emph{maximum edge participation} of a basis. This measures the maximum number of basis cycles that any edge participates in. Formally, for a cycle basis $B$ of $G = (V, E)$, we define:
\[
  \mu(B) = \max_{e \in E} \left| \{ C \in B : e \in C \} \right|.
\]
The motivation for minimizing the maximum edge participation comes from a well-known construction for quantum fault tolerance and is described further in Section \ref{sec:motivation}.

\subsection{Previous Work.}
Maximum edge participation is essentially identical to another quantity called the \emph{basis number}~\cite{schmeichel1981basis}, which has been studied under the notion of a \emph{$k$-fold} cycle basis—one in which every edge appears in at most $k$ basis cycles. The minimum such $k$ for which a k-fold cycle basis exists is called the basis number of the graph.
The basis number is the smallest $\mu(B)$ that can be obtained for any cycle basis $B$.

For certain families of graphs, tight bounds on $\mu$—the maximum edge participation in a cycle basis—are known, often implied by results on the \emph{basis number} of the graph. Planar graphs admit a cycle basis with $\mu \le 2$~\cite{mac1936combinatorial}, while complete graphs with $n \ge 5$ admit a basis with $\mu \le 3$~\cite{schmeichel1981basis}. Recent work has further connected $\mu$ to the genus of the graph: if a graph can be embedded on a surface of genus $g$, then it admits a cycle basis with $\mu = O(\log^2 g)$~\cite{miraftabsparse}. There is a trivial upper bound of $O(m)$ on the genus of a graph with $m$ edges, since each additional edge can be accommodated by adding a handle to the surface~\cite{milgram1977bounds}. This implies the existence of a cycle basis with maximum edge participation $\mu = O(\log^2 m)$ for any graph. In this paper, we restrict attention to simple graphs or multigraphs with constant edge multiplicity (as is sufficient for the applications mentioned in \cref{sec:motivation}), which reduces the bound to $\mu = O(\log^2 n)$.

On the other hand, a general lower bound holds for any connected graph: every cycle basis satisfies
\[
\mu \ge \frac{\text{girth}(G) \cdot (m - n + 1)}{m},
\]
as shown in~\cite{banks1982basis}.

For certain constant-degree expander graphs such as \textit{LPS expanders}~\cite{lubotzky1988ramanujan} which exhibit \textit{logarithmic girth} and \textit{genus} $O(n)$, it follows that the best general lower bound for $\mu$ will be
\[
\mu  = \Omega(\log n).
\]

\subsection{The Freedman–Hastings algorithm.} While theoretical bounds on $\mu(B)$ are known for some graph families, few algorithmic approaches explicitly target minimizing this quantity. A notable exception is the recent recursive algorithm of Freedman and Hastings~\cite{freedman2021building}, which on any constant-degree graph constructs a \textit{weakly fundamental} cycle basis with high probability and guarantees
\[
\mu = O(\log^2 n).
\]
A cycle basis is called \emph{weakly fundamental} if its cycles can be placed in a linear ordering such that each cycle contains at least one edge that does not appear in any later cycle in the ordering. Their method proceeds recursively, handling one of three cases at each step:

\begin{itemize}
    \item \textbf{Case 1}: If there is a degree-one vertex, remove it.

    \item \textbf{Case 2}: If there is a degree-two vertex \( v \) with neighbors \( x \) and \( y \):  
    \textbf{Case 2A}: If \( x \) and \( y \) are not connected by an edge, remove \( v \) and add an edge between \( x \) and \( y \).  
    \textbf{Case 2B}: If the edge \( (x, y) \) already exists (i.e., the three vertices form a triangle), add the triangle \( [x, v, y] \) to the cycle basis and remove \( v \).  
    These two subcases are illustrated in Figure~\ref{fig:Demonstration_Alg}.

    \item \textbf{Case 3}: If all vertices have degree at least three, find a short cycle of length \(O(\log n)\), add it to the basis, and randomly remove one of its edges from the graph.

\end{itemize}

\begin{figure}[htbp]
    \centering
    \includegraphics[width=0.49\textwidth]{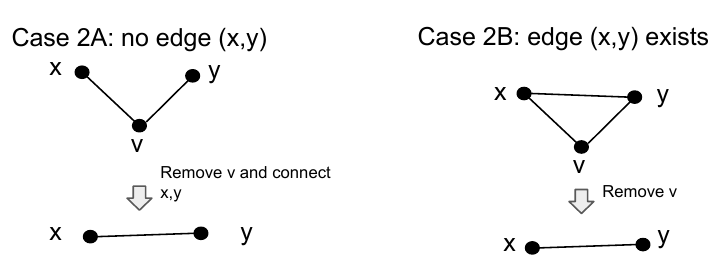}
    \caption{Illustration of \textbf{Case 2A} and \textbf{Case 2B} in the \textit{Freedman--Hastings} cycle basis algorithm. Vertex \( v \) has degree two, with neighbors \( x \) and \( y \). In \textbf{Case 2A}, we proceed recursively on the updated graph and after computing a cycle basis, reinsert v if any cycle basis contains $(x,y)$; in \textbf{Case 2B}, we add the triangle $[x,v,y]$ to the cycle basis and proceed recursively on the updated graph.
     }
    \label{fig:Demonstration_Alg}
\end{figure}

In the third case, a short cycle of length $O(\log n)$ is found using a breadth-first search (BFS). In this context, cross edges are defined as edges connecting a newly discovered vertex to an already visited one—i.e., all non-tree edges in the BFS tree. Starting from any vertex, any cross edge encountered during BFS, together with the paths from its endpoints to their lowest common ancestor, forms a cycle. The cycle formed by the first cross edge in the BFS tree is guaranteed to have length at most $2\lceil \log n \rceil$ (see Lemma 4.3 of~\cite{kavitha2009cycle}). In the \textit{Freedman-Hastings} algorithm, the BFS root is chosen uniformly at random, and they use the cycle generated by the first cross edge encountered.

However, the \( O(\log^2 n) \) bound achieved by this algorithm may be pessimistic for typical graphs. In~\cite{williamson2024low}, it was conjectured that for constant-degree random expanders, it might be possible to construct a cycle basis such that the expected maximum edge participation is only \( O(\log n) \). From this perspective, the \textit{Freedman--Hastings} construction can be seen as a worst-case bound, while more efficient average-case constructions may be possible with more refined cycle selection strategies. More broadly, the problem of minimizing maximum edge participation over all cycle bases remains a largely open combinatorial question, with practical implications in contexts such as quantum error correction and fault-tolerance~\cite{he2025extractors, hastings2021quantum}.

\subsection{Our contribution.}
In the \emph{Freedman–Hastings} algorithm, several design choices remain flexible and open to optimization, including:
\begin{itemize}
\item the starting vertex (root) for BFS;
\item the selection of the cross edge used to extract a short cycle;
\item the choice of which edge to remove from the cycle.
\end{itemize}
We explore heuristics that refine the edge and vertex selection rules, proposing a load-aware variant of the algorithm. The \emph{load} of an edge is defined as the number of basis cycles it has appeared in so far, while the \emph{load} of a vertex is the average load of its incident edges. To better reflect structural changes, we further define how load evolves during transformations: in Case 2A (where a new edge is inserted between two neighbors of a removed vertex), the load of the new edge is set to the maximum of the loads of the two removed edges. In Case 2B (where the connecting edge already exists), we increment the load of that existing edge by one.

Among the variants we tested, the following heuristic performed best empirically on random regular graphs:
\begin{itemize}
\item initialize each BFS from the vertex with the highest load;
\item select the first cross edge that forms a cycle containing the BFS root;
\item remove the \emph{highest-load} edge from the resulting cycle.
\end{itemize}


In Section~\ref{sec:motivation} we give an overview of how the cycle basis problem arises in quantum fault tolerance.
In Section~\ref{sec:experiments}, we present experimental results on random regular graphs to compare several of our best-performing heuristics. Section~\ref{sec:quantum} applies our algorithm to graphs derived from small quantum LDPC codes, demonstrating its practical relevance. The implementation and experimental scripts are available at~\cite{wang2025cyclebasiscode}. Finally, in Section~\ref{sec:theory}, we define a random process based on a balls-and-bins model to analyze and minimize the maximum load on a bin. This process serves as a proxy for the cycle basis problem on a random graph, capturing the distribution of load on edges under the simplifying assumption that cycles are random sets of edges. The edges chosen in a cycle is modeled by a random set of edges. The goal is to capture the best analysis that can be achieved without understanding the more complex distribution of edges chosen when cycles are selected in a real graph. Our analysis suggests that the maximum load of all of our heuristics will be $\Omega(\log^2 n)$.



%% file: motivation.tex
\section{Motivation from Quantum Fault Tolerance}  
\label{sec:motivation}
Due to the fragile nature of quantum systems, quantum error correction (QEC) is essential for building quantum computers capable of practical applications. QEC protects quantum information from environmental noise by redundantly encoding it using a quantum error-correcting code (QECC). In recent years, a particular class of QECCs known as \emph{quantum low-density parity-check} (QLDPC) codes~\cite{gottesman2013fault} has attracted significant attention for their potential to reduce the overhead associated with encoding and logical operations. 
A parity check is a measurement which measures whether an error has occurred. Low-density parity check, which only operate on a constant number of qubits, are easier to implement in a fault-tolerant way.
Recent works~\cite{panteleev2021quantum,breuckmann2021balanced,panteleev2022asymptotically,leverrier2022quantum} have proposed constructions of \emph{good} QLDPC codes—those with constant rate and linear distance.

However, a major challenge remains: how to perform logical gates fault-tolerantly on information encoded using these codes. A technique known as \emph{lattice surgery} was originally developed for implementing logical gates on surface codes~\cite{horsman2012surface}, and has since been generalized to broader classes of QLDPC codes~\cite{cohen2022low,cross2024linear,williamson2024low,he2025extractors}. The core goal of lattice surgery is to enable fault-tolerant measurements of logical operators, which in turn allows the implementation of an almost universal gate set—the Clifford group. This set, when supplemented by \emph{magic state distillation}~\cite{bravyi2005universal}, yields a universal gate set for quantum computation~\cite{litinski2019game}.

Every QLDPC code has a Tanner graph representation, where physical qubits and checks are nodes, and edges indicate which checks act on which qubits. We define the \emph{weight} of a check to be the number of physical qubits it touches, and the \emph{degree} of a qubit to be the number of checks connected to it.

Without loss of generality, we focus on CSS codes, which involve only two types of checks: $Z$-checks and $X$-checks. (In general, a check can involve both $X$ and $Z$ operators.) These two types of checks correspond to the two common types of errors in quantum systems: $X$-errors (bit-flip errors) and $Z$-errors (phase-flip errors). For CSS codes, the Tanner graph is a tripartite graph—one part for qubits, one for $X$-checks, and one for $Z$-checks—analogous to the bipartite structure in classical LDPC codes.

Correspondingly, there are two types of logical measurements: Pauli $X$ measurements and Pauli $Z$ measurements. In this discussion, we focus on Pauli $X$ measurements; the procedure for Pauli $Z$ measurements is analogous. A Pauli $X$ measurement can be viewed as a measurement on a subset of physical qubits. However, such measurements often have high weight, meaning they involve many physical qubits. Performing high-weight measurements directly is challenging in a fault-tolerant manner.

The goal of lattice surgery is to decompose such high-weight measurements into a sequence of lower-weight, fault-tolerant measurements. The basic idea is to introduce an ancillary system (which can also be thought of as another QLDPC code) that is coupled to the original code on which we wish to perform the measurement. By measuring all the $X$-checks in the ancilla system, one can effectively obtain the outcome of the desired high-weight Pauli $X$ measurement.

There have been several approaches for constructing the ancilla system~\cite{cohen2022low,cross2024linear,williamson2024low}. These constructions share a common starting point: they begin with the Tanner graph and focus on the induced subgraph of the support (i.e., the set of qubits involved in the Pauli $X$ measurement.)
The cycle basis problem comes out of the Williamson et. al. construction.
Although the details of the construction are beyond what we can cover 
here, the result of the construction is 
 a new graph based on the induced subgraph of the Tanner graph in which 
edges represent physical qubits, vertices correspond to $X$-checks, and $Z$-checks are defined via a cycle basis $B$ of the graph $G$, which together define the ancilla system.

Each basis cycle corresponds to a $Z$-check in the ancilla system, and its length determines the weight of that check. Long cycles lead to high-weight checks, which can violate the LDPC property of the original code. However, there is a straightforward method to resolve this: cellulation—the process of subdividing long cycles into multiple short cycles by adding auxiliary edges, thereby restoring LDPC properties.

In contrast, maximum edge participation corresponds to the maximum qubit degree in the ancilla system—that is, the number of $Z$-checks attached to a given qubit. High qubit degrees can violate the LDPC property of the code, and unlike the case of high-weight checks, there is no comparably low-cost  fix such as cellulation. The overhead of the construction is then dominated by the cost of reducing the maximum edge participation.

%% file: experiments.tex
\section{Experimental Comparison of Different Heuristics on Random Regular Graphs.}
\label{sec:experiments}

In this section, we present four different variants of the original Freedman–Hastings algorithm 
by varying the three key design choices listed in Table~\ref{tab:variants}.
We will refer to the Freedman–Hastings algorithm  as Version $0$.

\begin{table*}[htbp]
\centering
\caption{Variants of the Freedman–Hastings algorithm with different design choices (Version 0 is the original version).}
\label{tab:variants}
\begin{tabular}{|c|c|c|c|}
\hline
\textbf{Variant Name} & \textbf{BFS Root Selection} & \textbf{Cross Edge Selection} & \textbf{Edge Removal Strategy} \\
\hline
Version 0 & Random vertex & First encountered & Random edge \\
\hline
Version 1 & Random vertex & First encountered & Max-load edge \\
\hline
Version 2 & Max-load vertex & First encountered & Max-load edge \\
\hline
Version 3 & Max-load vertex & First forming cycle with root & Max-load edge (prefer incident to root)  \\
\hline
Version 4 & Max-load vertex & First forming cycle with root & Load-weighted probabilistic removal \\
\hline
\end{tabular}
\end{table*}

The only difference between Version 1 and Version 0 is the edge removal strategy: while Version 0 removes a random edge from the selected cycle, Version 1 removes the edge with the highest load (breaking ties randomly). This strategy helps reduce long-term edge congestion by preferentially eliminating the most heavily used edges early in the recursion.

The difference between Version 2 and Version 1 lies in the BFS root selection. While Version 1 selects the BFS root randomly, Version 2 always starts from the vertex with the highest load. Although this factor may seem less impactful than edge removal, empirical results show that prioritizing the max-load vertex leads to better performance. Intuitively, a high-load vertex is likely incident to multiple heavily loaded edges. By starting BFS from such a vertex, the resulting cycle has a higher chance of containing one of these high-load edges, which can then be removed using the same greedy strategy.

Version 3 improves upon Version 2 by refining both the cross edge selection and edge removal strategies. Instead of selecting the first encountered cross edge during BFS, we restrict the choice to those that yield a cycle containing the BFS root. This ensures that the starting vertex—which is chosen based on maximum load—remains in the cycle, increasing the chance of targeting heavily used edges.

Additionally, when multiple edges in the cycle have the same maximum load, we prioritize removing an edge that is incident to the BFS root. This tie-breaking rule is particularly effective in the later stages of the algorithm when the graph becomes sparse. For instance, in a 3-regular graph, removing one edge via a short cycle (as in Case 3) typically creates two degree-two vertices—the endpoints of the removed edge—which will each trigger additional reductions via Case 2; see Figure~\ref{fig:Demonstration_exp} for an illustrative example. 
If the BFS root lies on the cycle and has high load, this refinement increases the likelihood of further removing overloaded edges connected to it, thus improving overall balance in the cycle basis.

\begin{figure}[htbp]
    \centering
    \includegraphics[width=0.49\textwidth]{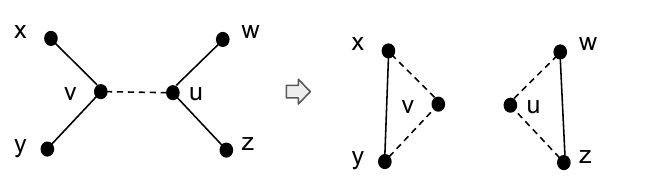}
    \caption{Illustration of how Case 2 introduces additional reductions in a 3-regular graph following Case 3. Suppose edge \( (u, v) \) is removed as part of a short cycle in Case 3. After removal, both \( u \) and \( v \) become degree-two vertices. By Case 2A, edges \( (x, v) \) and \( (v, y) \) are replaced by \( (x, y) \), and the load on \( (x, y) \) is set to the maximum of the loads on \( (x, v) \) and \( (v, y) \). This gives an advantage only if all 3 edges incident to $v$ has high load. A similar reduction applies to the pair \( (w, z) \) resulting from the vertex \( u \).}
    \label{fig:Demonstration_exp}
\end{figure}

Finally, Version 4 introduces a more probabilistic edge removal strategy. Instead of deterministically removing the highest-load edge, we remove edge $i$ from the selected cycle $C$ with probability
$$2^{L(i)} / \sum_{j \in C} 2^{L(j)},$$
where $L(i)$ denotes the load of edge $i$. This softmax-style weighting introduces randomness into the algorithm, which may help avoid worst-case scenarios that can trap greedy variants. However, empirical results suggest that Version 4 does not outperform Version 3 on average.

We evaluate our algorithms on random regular graphs, which are known to be asymptotically close to Ramanujan graphs—optimal spectral expanders—according to the result of Friedman~\cite{friedman2003proof}. This motivates our choice, as most random 
d-regular graphs are constant-degree expanders, resembling the structure of graphs that arise in real quantum LDPC codes.

To generate these graphs, we use Python’s NetworkX package, which internally implements the configuration model of Steger and Wormald~\cite{steger1999generating}. Specifically, we test our heuristics on random 3-regular and 8-regular graphs across a range of sizes. We chose these graphs to roughly mimic the structure of near-term quantum codes, which typically do not exhibit very high connectivity, check weight, or qubit degree. For each graph size, we generate a number of random instances (as detailed in Table~\ref{tab:trialcounts}) and apply all five algorithmic variants. 

\textit{Running time analysis.}
All variants (Versions 1–4) have the same asymptotic running time as the
original Version 0 (the \textit{Freedman–Hastings} algorithm). Each iteration removes at least one edge and may require
operations such as selecting the max-load vertex, performing a BFS, or identifying the
max-load edge within a short cycle—all of which take at most $O(m)$ time. Since there
are at most $O(m)$ such iterations throughout the recursion, the total worst-case
running time remains
$
O(m^2) \quad (\text{equivalently } O(n^2)\text{ for constant-degree graphs}).$  Since the runtime of each algorithm grows roughly quadratically with graph size, we reduced the number of trials for larger graphs to keep the total computation time manageable. All experiments were performed on a desktop workstation (Intel\textsuperscript{\textregistered} Core\textsuperscript{TM} i7-14700F CPU, 16\,GB RAM, 64-bit Windows~11). Completing all trials required approximately three days of total wall-clock time. 

\begin{table}[htbp]
\centering
\caption{Number of random \(d\)-regular graphs generated per size \(n\), for each degree \(d\).}
\label{tab:trialcounts}
\begin{tabular}{|c|c|c|}
\hline
\textbf{Graph size \(n\)} & \textbf{Trials for \(d = 3\)} & \textbf{Trials for \(d = 8\)} \\
\hline
32     & 4000 & 4000 \\
64     & 4000 & 4000 \\
128    & 4000 & 2000 \\
256    & 2000 & 1000  \\
512    & 1000  & 500  \\
1024   & 500  & 250 \\
2048   & 250  & 100   \\
4096   & 100   & 50   \\
8192   & 50   & 20   \\
16384  & 20   & --   \\
\hline
\end{tabular}

\vspace{0.5em}
\raggedright
\footnotesize{Note: No experiments were run for \(d = 8\) at \(n = 16384\) due to runtime constraints.}
\end{table}

We collectively visualize the results across all graph sizes and degrees using the boxplot from multiple trials, as shown in Figure~\ref{fig:boxplot-participation}. Empirically, all of our algorithmic variants demonstrate significant improvement over the original \textit{Freedman–Hastings} algorithm, suggesting that the max-load edge removal strategy is the primary factor driving these performance gains.

\begin{figure*}[htbp]
    \centering
    \includegraphics[width=\textwidth]{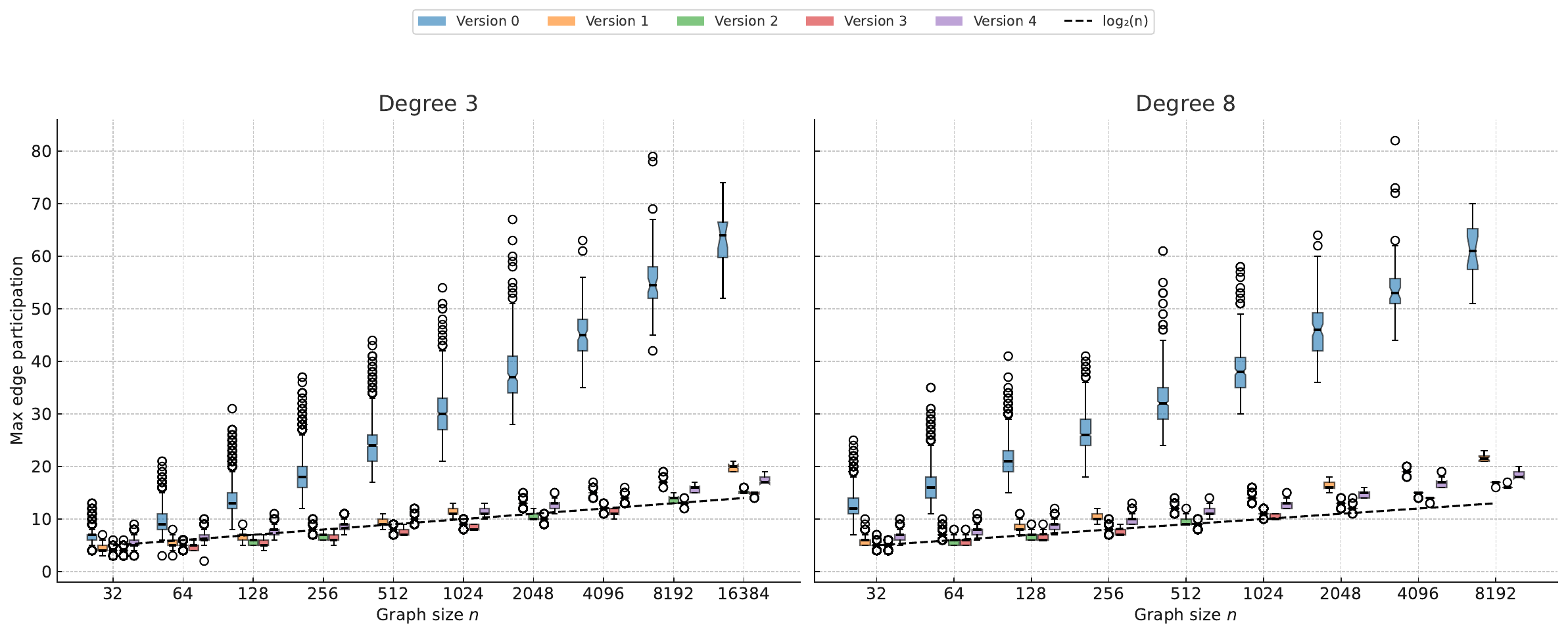}
    \caption{
    Boxplot comparison of the maximum edge participation across five algorithmic variants of the Freedman--Hastings algorithm, evaluated on random \( d \)-regular graphs. The left panel corresponds to \( d = 3 \), and the right panel to \( d = 8 \). The boxplot aggregates the results from thousands of random graph instances (see Table~\ref{tab:trialcounts}). Each variant is color-coded, and the black dashed curve in both panels shows the baseline \( \log_2(n) \) scaling for reference. The boxplots summarize the distribution of the maximum edge participation over all trials: the central line indicates the median, the box bounds the interquartile range (IQR), and the whiskers extend to 1.5×IQR. Outliers beyond that range are represented by small circles.
    }
    \label{fig:boxplot-participation}
\end{figure*}

Among the variants, Version 3 consistently achieves the lowest maximum edge participation on average, followed by Version 2, then Version 4, and finally Version 1. This ordering confirms that both the BFS root selection and cross edge selection strategies contribute meaningfully to reducing edge congestion—beyond the effect of edge removal alone.

To further understand the asymptotic behavior of the different algorithmic variants, we plot in Figure~\ref{fig:version0-vs-3-fit-ratio} the ratio between the actual median maximum edge participation and fitted functions of the form \(c \log_2(n)\) and \(c \log_2^2(n)\), where the coefficient \(c\) is independently optimized for each curve. This normalization removes the leading-order growth rate and highlights the deviation from ideal logarithmic or polylogarithmic scaling. 

\begin{figure*}[htbp]
    \centering
    \includegraphics[width=\linewidth]{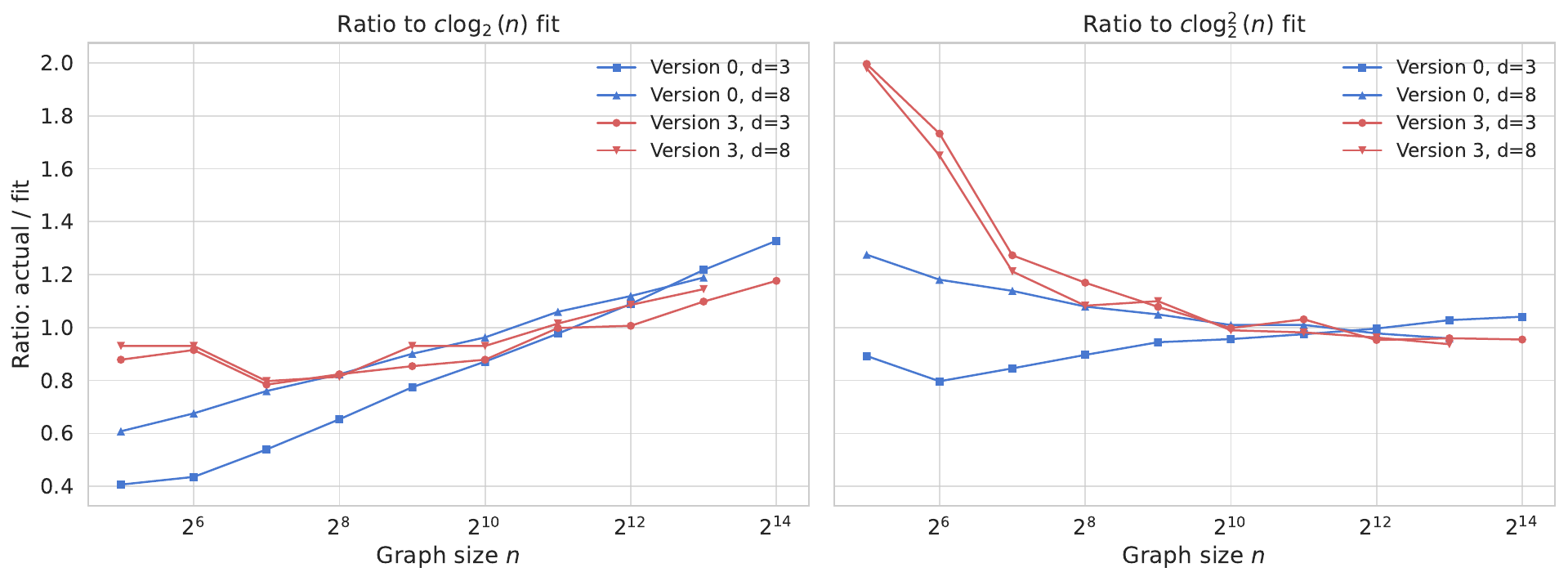}
    \caption{
    Comparison of the ratio between the actual median maximum edge participation and fitted functions \(c\log_2(n)\) (left) and \(c\log_2^2(n)\) (right), for Versions 0 and 3 of the algorithm. Each curve is normalized by its own fitted coefficient \(c\), independently optimized per curve. Solid lines are used throughout, with marker shape distinguishing the degree: squares (degree 3) and triangles (degree 8) for Version 0 (blue), and circles (degree 3) and inverted triangles (degree 8) for Version 3 (red). A flat curve near 1 indicates better scaling agreement with the fitted asymptotic model.
    }
    \label{fig:version0-vs-3-fit-ratio}
\end{figure*}

We include results for Version 0 and Version 3 of the algorithm across both degree-3 and degree-8 random regular graphs. All curves use solid lines, with different marker shapes distinguishing the vertex degree. In the left plot, we observe that the ratios for both Version 0 and Version 3 trend upward, indicating super-logarithmic scaling for both versions. In contrast, the right plot shows that the ratios for Version 0 remain relatively flat, whereas those for Version 3 initially decrease but eventually level off. This suggests that Version 3 offers only a constant-factor improvement in scaling relative to Version 0, rather than a fundamentally different asymptotic behavior. This observation is also consistent with the lower bound established for Version 3 in Section~\ref{sec:theory}.



%% file: quantum.tex
\section{Experiment on Graphs from Small Quantum Codes}
\label{sec:quantum}

In this section, we evaluate our algorithm on graphs arising from practical quantum codes. Specifically, we use quantum radial codes introduced by Scruby et al.~\cite{scruby2024high}, which are based on the lifted product construction~\cite{panteleev2021quantum} applied to certain quasi-cyclic classical codes. The lifted product is one of the first constructions known to yield good quantum LDPC codes, and it is also effective for generating near-term deployable codes.

Quantum radial codes exhibit promising features for fault-tolerant quantum computing. In particular, they offer comparable error suppression to surface codes while requiring roughly five times fewer physical qubits. Moreover, they have been observed to support single-shot decoding~\cite{bombin2015single}, making them strong candidates for near-term implementations.

We test our algorithm on graphs arising from two specific quantum radial codes presented in~\cite{scruby2024high}: a $[[90,8,10]]$ code and a $[[352,18,20]]$ code. Here, a $[[n,k,d]]$ code refers to a quantum code with $n$ physical qubits, $k$ logical qubits, and code distance $d$. These codes have parameters comparable to those of the Gross code and the Double Gross code~\cite{bravyi2024high}, but have received relatively less attention in the literature.

In~\cite{scruby2024high}, the authors provide a canonical basis for all Pauli-$X$ and Pauli-$Z$ logical operators. In our experiments, we focus on the Pauli-$X$ basis (note that the procedure for Pauli-$Z$ measurements would be essentially identical). Following the logical operator measurement construction in~\cite{williamson2024low}, we extract, for each Pauli-$X$ basis operator, an induced subgraph from the Tanner graph representation of the code. Specifically, each intersecting $Z$-check becomes an edge in the induced graph; in the two quantum radial codes we study, every $Z$-check intersects each logical operator on exactly two qubits, so this conversion is unambiguous.

For the $[[90,8,10]]$ code, this yields 8 graphs (one for each logical $X$ operator), and for the 
$[[352,18,20]]$ code, we obtain 18 graphs. In general, to ensure distance preservation during lattice surgery, one may need to augment these graphs with additional edges to boost their Cheeger constants to at least 1, as this is a sufficient (but not necessary) condition for maintaining code distance. However, in our case, we observe that for the smaller code, all 8 graphs already have Cheeger constant exactly 1. For the larger code, 12 out of the 18 graphs have Cheeger constant $\geq 1$, 3 have approximately 0.91, and the remaining 3 have values near 0.73.

Given these relatively high expansion properties, we directly test our algorithms on these induced subgraphs without further augmentation, expecting the results to closely reflect those obtained in a full lattice surgery simulation.

The graphs induced from the smaller code each have 10 vertices—comparable to the code distance—and are 3-regular. Two representative examples 
are shown in Figure~\ref{fig:sample-graphs_small}. For the larger code, the induced graphs each have approximately 22 vertices and are 4-regular; two examples are shown in Figure~\ref{fig:sample-graphs_large}.

\begin{figure}[htbp]
    \centering
    \begin{minipage}[t]{0.48\textwidth}
        \centering
        \includegraphics[width=\linewidth]{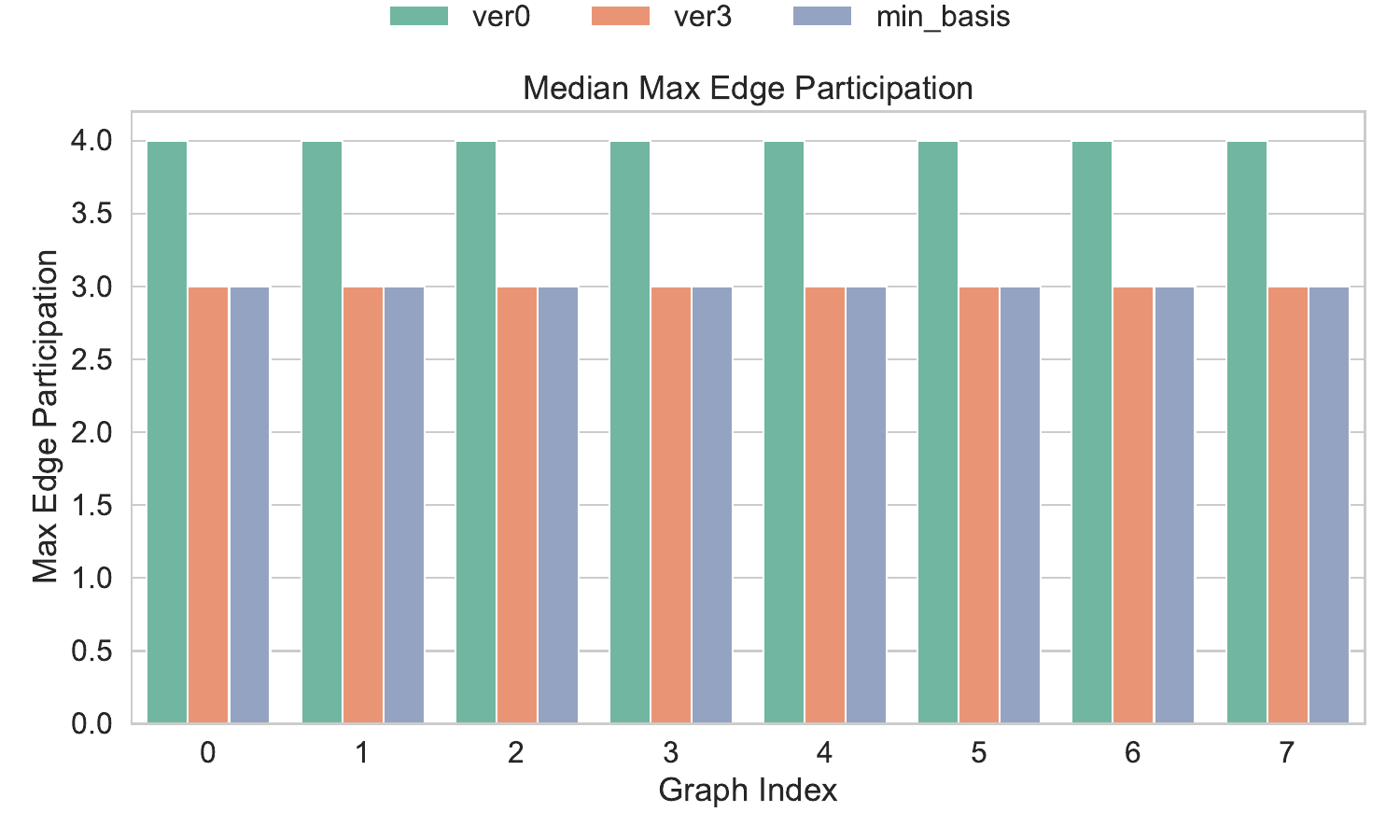}
        \caption{Median maximum edge participation across algorithmic variants for the graphs arising from} the $[[90,8,10]]$ code.
        \label{fig:median-barplot_small}
    \end{minipage}
    \hfill
    \begin{minipage}[t]{0.48\textwidth}
        \centering
        \includegraphics[width=\linewidth]{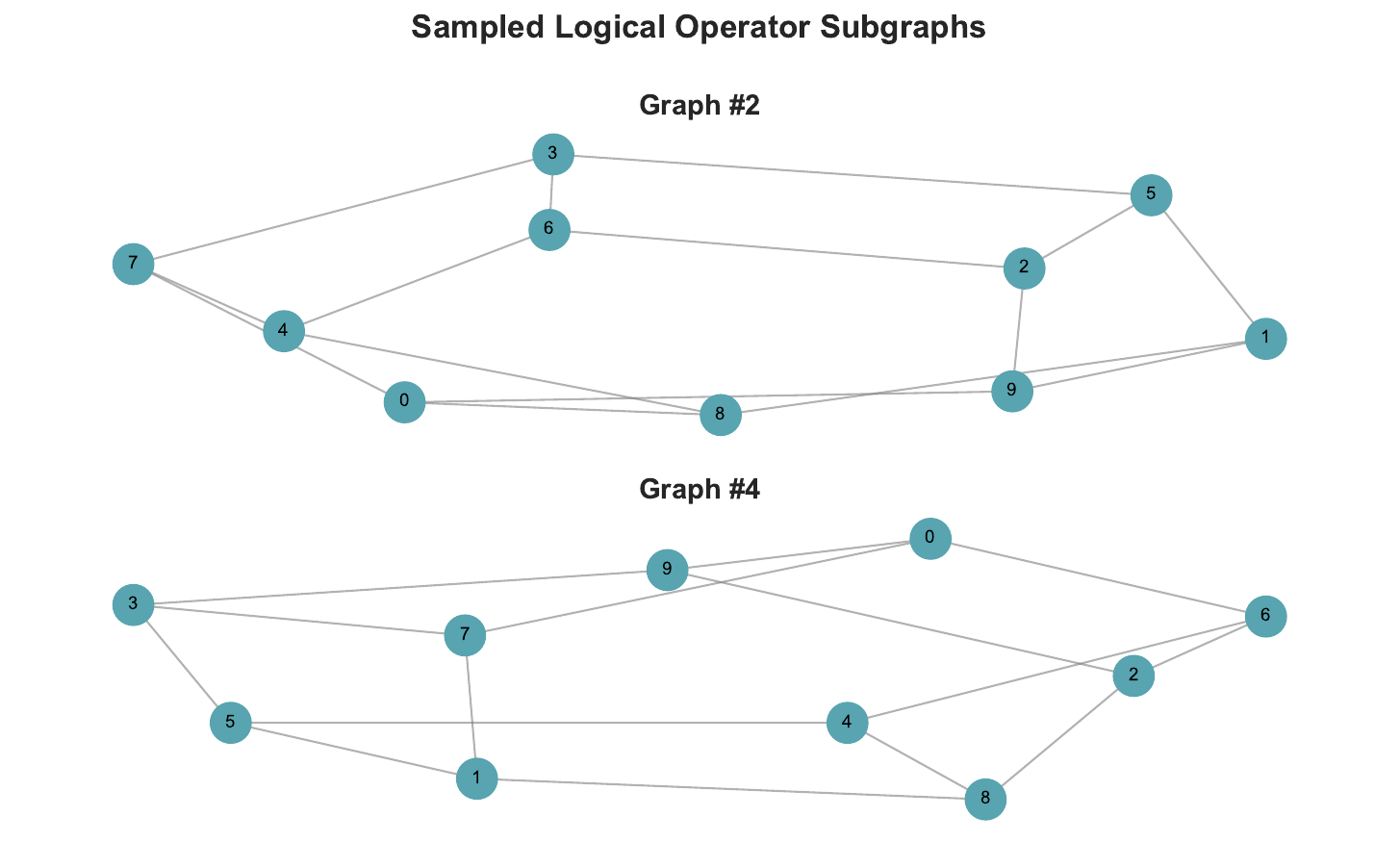}
        \caption{Two example logical-operator--induced subgraphs from the dataset. Each edge represents a qubit in the ancilla system used for lattice surgery, and each vertex corresponds to an $X$-check. The choice of cycle basis determines the associated set of $Z$-checks.}
        \label{fig:sample-graphs_small}
    \end{minipage}
\end{figure}

\begin{figure*}[htbp]
    \centering
    \begin{minipage}[t]{0.48\textwidth}
        \centering
        \includegraphics[width=\linewidth]{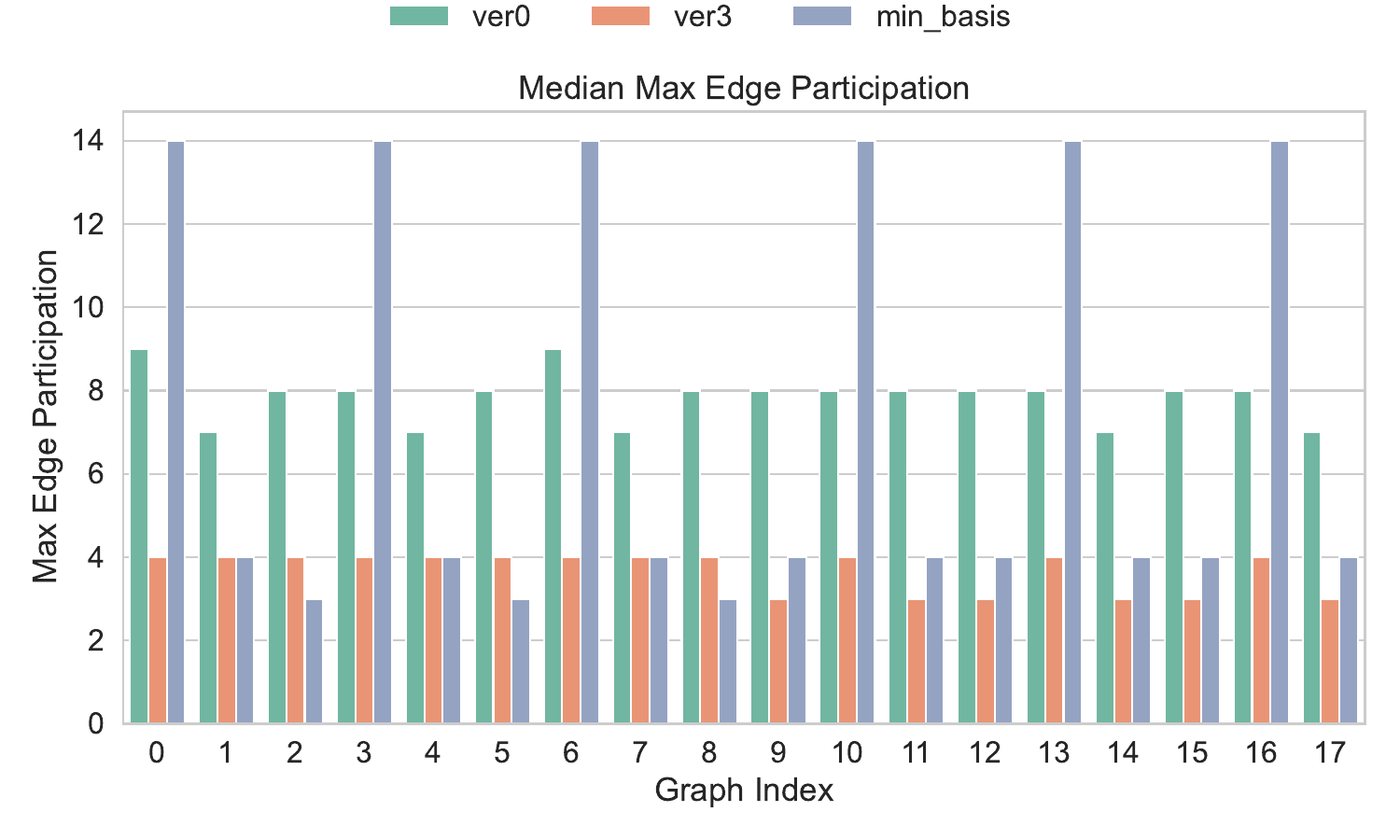}
        \caption{Median maximum edge participation across algorithmic variants for the graphs arising from} the $[[352,18,20]]$ code.
        \label{fig:median-barplot_large}
    \end{minipage}
    \hfill
    \begin{minipage}[t]{0.48\textwidth}
        \centering
        \includegraphics[width=\linewidth]{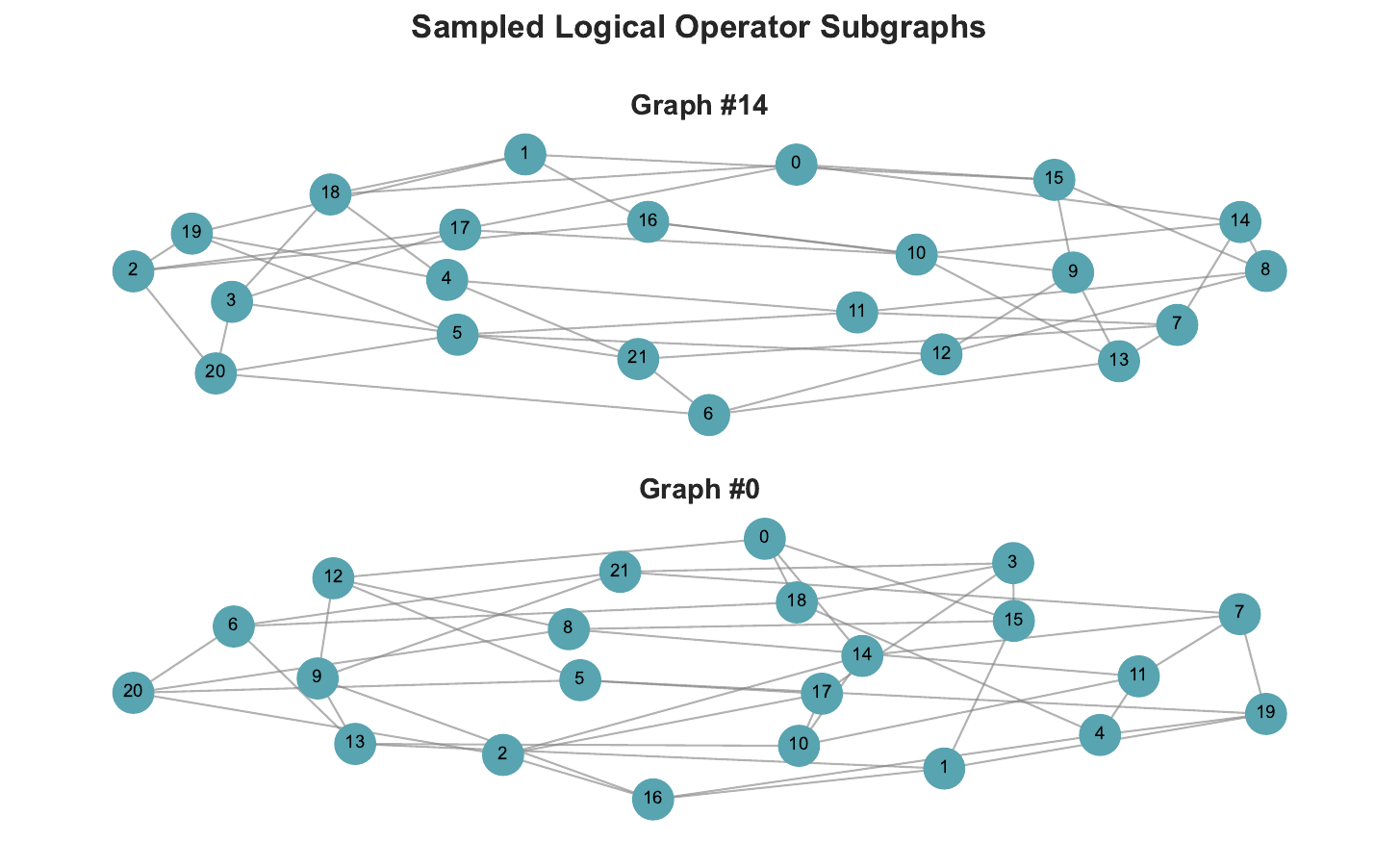}
        \caption{Two example logical-operator--induced subgraphs from the dataset. Each edge represents a qubit in the ancilla system used for lattice surgery, and each vertex corresponds to an $X$-check. The choice of cycle basis determines the associated set of $Z$-checks.}
        \label{fig:sample-graphs_large}
    \end{minipage}
\end{figure*}
The algorithms we test include Version 0 (the original Freedman–Hastings algorithm), Version 3 (our best-performing heuristic), and min-basis (a minimum-weight cycle basis algorithm~\cite{kavitha2008algorithm}). We include the minimum-weight cycle basis as a baseline because, in the context of lattice surgery, the length of the basis cycles can contribute to overhead—though typically not as severely as maximum edge participation.



The performance of the three algorithms is shown in Figures~\ref{fig:median-barplot_small} and~\ref{fig:median-barplot_large}, corresponding to the smaller and larger codes, respectively. For each graph derived from the smaller code, we run each algorithm 100 times and report the median result; for the larger code, each graph is evaluated over 500 trials.

We observe that for the smaller code, the performance differences among the algorithms are relatively minor. However, for the larger code, the performance gap becomes much more pronounced. Version 3 consistently performs well, achieving roughly a 50\% reduction in maximum edge participation compared to Version 0. The min-basis algorithm occasionally matches Version 3 but can also yield the worst performance—this is expected, as min-basis is a deterministic algorithm optimized for cycle length rather than maximum edge participation.

%% file: theory.tex
\section{A balls-into-bins model as a proxy}\label{sec:theory}

We consider a balls-into-bins model as a proxy for the behavior of our algorithms on 3-regular graphs. In this model, each bin corresponds to an edge in the graph, and each selected cycle is a subset of   $k$ bins.
The load on each edge is represented by placing a ball in each of the selected bins. Then three bins are removed, representing the edges
that are removed from the graph. Note that when the graph is exactly $3$-regular, three edges are removed in each iteration of Case $3$, as illustrated in Figure \ref{fig:Demonstration_exp}. The first edge removal comes from Case $3$, where an edge is removed from the chosen cycle. This creates two degree-$2$ vertices, each of which results in the removal of another edge from Case $2$.

The goal of the model is to understand the extent to which we can analyze the asymptotic behavior of our best variant without using detailed graph-theoretic properties. The variants of the cycle basis algorithm discussed in the previous sections allow for selecting a cycle that includes up to three heavily loaded edges. The rest of the edges in the cycle depend on the cross edges encountered in constructing the BFS tree.
In the balls-to-bins model, we model the selection of a cycle as selecting
the three most heavily loaded edges/bins, and the rest of the $k-3$ bins are selected at random. Then the three most heavily loaded bins are removed from the system, representing the removal of those edges from the graph. This model is favorable to the algorithm by allowing the algorithm to select and remove the three most heavily loaded bins in the entire system. 
In a real graph, the three most heavily loaded edges might not even be located on a single short cycle.
The selection of the rest of the edges is simplified to be a random set. Intuitively, getting a better analysis of the behavior of our cycle basis algorithms will require a more sophisticated understanding of the distribution of the edges participating in each selected cycle. 

There are a few other ways in which our proxy process differs from cycle basis selection in a real graph.
The proxy process assumes that the cycle selected will have length exactly $k = 2\left\lceil \log_2 n \right\rceil$ in each round.
The versions which select the cycle arising from the first cross edge encountered may yield shorter cycles, resulting in less load on the edges.
However, in a random graph, the first cross edge will appear when the BFS tree has approximately 
$\sqrt{n}$ vertices, resulting in a cycle of length $\Omega(\log n)$. Intuitively, if $k$ vertices have been discovered, each new edge has probability about $k/n$ of connecting to an already visited vertex, so the expected number of cross edges is $\Theta(k^2/n)$. When $k$ reaches $\sqrt{n}$, this expectation becomes constant, marking the point where the first cycle typically appears.
Thus, in the absence of a more involved model for how the graph evolves over time, we do not expect to be able to leverage this effect to obtain an asymptotically better  analysis. 
The proxy process also does not capture vertex removals from Cases $1$ and $2B$. However, our numerical results suggest that these are relatively infrequent events and therefore do not contribute significantly to edge load as shown in \cref{Sec: B}. 

The balls-and-bins process is called \emph{Process 1} and is given in Algorithm \ref{process1}.
However, Process~1 remains analytically challenging. To facilitate analysis, we define a further simplification,  which still preserves the essential features needed for bounding load growth. 
In the simplified process, we define an \emph{epoch} as a sequence of iterations in which the number of buckets $m$ decreases by a factor of $2$. During each iteration within an epoch, we add a ball to each of $k$
randomly chosen bins and decrease $m$, the counter for the number of buckets, by $3$.
However, the buckets are not actually removed until the end of an epoch, at which point a new epoch begins. The delayed-removal process is called \emph{Process 2} and is given in Algorithm \ref{process2}.



\begin{algorithm}
\caption{Process 1}
\begin{algorithmic}
\label{process1}
\REQUIRE $M  > 0$, the initial number of buckets.
\STATE $m \gets M$
\STATE $n \gets \tfrac{2}{3}M$
\STATE All $m$ buckets start empty.
\STATE $m_{\min}$ is some fixed constant.
\WHILE{$m > m_{\min}$ }
\STATE $k \gets 2\left\lceil \log_2 n \right\rceil$
\STATE Select the $3$ buckets with the highest  load.
  \STATE Select $k - 3$ additional buckets  at random.
\STATE  Add one ball to each of the $k$ selected buckets.
\STATE Remove the three most heavily loaded buckets.
\STATE $m \gets m - 3$
\STATE $n \gets \tfrac{2}{3}m$
\ENDWHILE
\end{algorithmic}
\end{algorithm}

\begin{algorithm}
\caption{Process 2 (delayed removal)}
\begin{algorithmic}
\label{process2}
\REQUIRE $M  > 0$, the initial number of buckets.
\STATE $m \gets M$
\STATE $j \gets 1$
\STATE All $m$ buckets start empty.
\STATE $m_{\min}$ is some fixed constant.
\WHILE{$m > m_{\min}$ }
\STATE $m_{old} \gets m$
\STATE $n \gets \tfrac{2}{3}m$.
\STATE $k \gets 2\left\lceil \log_2 n \right\rceil$
\WHILE{$m > \frac{M}{2^{j}}$}
  \STATE Select $k $  buckets  at random.
\STATE  Add one ball to each of the $k$ selected buckets.
\STATE $m \gets m-3$
\ENDWHILE
\STATE $r \gets (m_{old} - m)$
\STATE Remove the $r$ most heavily loaded buckets.
\STATE $j \gets j+1$
\ENDWHILE
\end{algorithmic}
\end{algorithm}




This delayed-deletion model allows additional load to accumulate before pruning. As a result, more balls are discarded during each epoch compared to Process~1. Therefore, any lower bound established for Process~2 also applies to Process~1, and thus to the original algorithm. In \cref{Sec: C}, we provide a formal proof that Process~2 lower-bounds Process~1 if we select \((k - 3)\) buckets instead of \(k\) buckets in Process~2; note that this modification does not affect the asymptotic lower bound for Process~2 discussed below.

In what follows, we prove the following theorem, which provides a lower bound on the maximum load $L_{\max}$ in Process~2. All logairthms are base $2$, unless stated otherwise.
\begin{theorem}
\label{theorem:main}
For any fixed \(c \in (0, 0.16)\), process~2 satisfies
\[
   L_{\max}
   \;=\;
   \Omega\!\bigl(\log^2 M \bigr)
   \quad\text{with high probability}.
\]
Explicitly,
\[
   L_{\max}
   \;\ge\;
   \frac{c^2}{6} \left(1 - \frac{c}{2} \right) \log^2 M
   \quad\text{w.h.p.}
\]
\end{theorem}

Before proving the theorem, we develop several lemmas. Define $m_j = m / 2^j$, $n_j = \tfrac{2}{3}m_j$, and $k_j = 2\left\lceil \log_2 n_j \right\rceil$. During epoch~$j$, we begin with $m_j$ buckets and perform ball tossing for $m_j / 6$ rounds. In each round, we toss $k_j$ balls into $k_j$ distinct buckets chosen uniformly at random \emph{without replacement}. No deletions occur during the rounds.

Consider what happens at one epoch first. Fix an epoch j and the probability for any fixed bucket being chosen a single round is $k_j/m_j$. The incremental load
of any bucket before removing buckets at the end of  epoch~$j$ is
\[
   X_j \;\sim\;
   \operatorname{Binomial}\!\Bigl(\tfrac{m_j}{6},
                                   \tfrac{k_j}{m_j}\Bigr),
   \mu_j := \mathbb{E}[X_j]
          = \frac{m_j}{6}\,\frac{k_j}{m_j}
          = \frac{k_j}{6}.
\]

Fix a constant \(c\in(0,0.16)\) and call a bucket  
\emph{bad in epoch $j$} if the incremental load satisfies  
\[
   X_j \leq c\,\mu_j .
\]
During epoch~$j$ let $B_j$ be the number of \textit{bad} buckets. 
We state the following two lemmas, whose proofs are provided in \cref{Sec: A}.

\begin{lemma}[Few bad buckets in one epoch]\label{lem:few}
Fix $c \in (0, 0.16)$ and set
$\alpha = \frac{(1-c)^2}{6 \ln 2}$.
During epoch~$j$, let $B_j$ be the number of buckets whose incremental
load is at most $c\mu_j$. Then
\[
  \Pr\left(B_j \geq \left(\tfrac{3}{2}\right)^\alpha m_j^{1 - \frac{\alpha}{2}}\right) \leq m_j^{-\frac{\alpha}{2}}.
\]
\end{lemma}

\begin{lemma}[Few bad buckets over $L$ epochs]
\label{lem:few_total}
Fix $c \in (0, 0.16)$ and define $\alpha = \frac{(1-c)^2}{6\ln 2}$. 
Let \( L = \left\lceil \frac{c}{2} \log_2 M \right\rceil \). Then with probability $1 - o(1)$, the total number of buckets that are bad in \emph{any} of the first $L$ epochs is at most
\[
   C(c)\,M^{1 - \frac{\alpha}{2}},
   \qquad\text{where}\qquad
   C(c) := (\tfrac{3}{2})^\alpha \cdot \frac{1}{1 - 2^{-(1 - \frac{\alpha}{2})}}.
\]
\end{lemma}

We are now ready to prove Theorem \ref{theorem:main}.
The basic idea is that for the chosen value of $L$ (which is  $\Omega(\log n)$), 
with high probability the number of buckets that have ever gone bad in 
$L$ rounds is less than the total number of buckets remaining. Thus, with high probability, there will be a bucket remaining which was never bad in the first $L$ rounds.
A bucket with this property received $\Omega(\log n)$ incremental load in each of  the $L$ rounds,
resulting in a bucket with total load $\Omega(\log^2 n)$.

\begin{proof} 
(\cref{theorem:main})
Since each deleted bucket has a unique ancestor among the original $M$ buckets, \cref{lem:few_total} bounds how many buckets \emph{ever} go bad. Provided that $\alpha(c) > c$, which holds for every $c < 0.16$, we have
\[
M^{1 - \frac{\alpha}{2}} = o\left(M^{1 - \frac{c}{2}}\right),
\]
and hence
\[
C(c)\,M^{1 - \frac{\alpha}{2}} < m_L
\]
for sufficiently large $M$, forcing the existence of at least one bucket that is never bad in any epoch $j < L$.

For such a never-bad bucket, its load in epoch~$j$ satisfies
\[
X_j \ge c\,\mu_j \ge \frac{c}{3} \log_2\left( \tfrac{2}{3} m_j \right),
\]
using the lower bound $\mu_j \ge \frac{1}{3} \log_2\left( \tfrac{2}{3} m_j \right)$ from earlier.

Summing over \(L = \left\lceil \frac{c}{2} \log_2 M \right\rceil\) epochs, we obtain a lower bound on the total load on that bucket:
\begin{align*}
L_{\max}
  &\ge \frac{c}{3} \sum_{j=0}^{L-1} \log_2\left( \tfrac{2}{3} m_j \right) \\
  &= \frac{c}{3} \sum_{j=0}^{L-1} \left( \log_2 M - j + \log_2 \tfrac{2}{3} \right) \\
  &\ge \frac{c}{3} \cdot L \cdot \left( \log_2 M - L + 1 + \log_2 \tfrac{2}{3} \right),
\end{align*}
where we use that $\log_2 m_j = \log_2 (M \cdot 2^{-j}) = \log_2 M - j$ and that $\log_2\left( \tfrac{2}{3} m_j \right)$ is decreasing in $j$, so each term is at least the last one.

Substituting the bounds \( \frac{c}{2} \log_2 M \le L \le \frac{c}{2} \log_2 M + 1 \), we control both terms in the product:
\begin{align*}
L_{\max}
  &\ge \frac{c}{3} \cdot L \cdot \left( \log_2 M - L + 1 + \log_2 \tfrac{2}{3} \right) \\
  &\ge \frac{c}{3} \cdot \frac{c}{2} \log_2 M \cdot \left( \log_2 M - \left( \frac{c}{2} \log_2 M + 1 \right) + 1 \right) \\
  &= \frac{c}{3} \cdot \frac{c}{2} \log_2 M \cdot \left(1 - \frac{c}{2} \right) \log_2 M \\
  &= \frac{c^2}{6} \left(1 - \frac{c}{2} \right) \log_2^2 M.
\end{align*}

Thus,
\[
L_{\max} = \Omega(\log^2 M),
\]
completing the proof of \cref{theorem:main}.
\end{proof}
\newpage

%% file: appendix.tex
\appendix\
\section{Proof of the Lemmas} \label{Sec: A}
\begin{proof}
(\cref{lem:few})
For a single bucket, the multiplicative Chernoff lower tail bound
\[
  \Pr[X \leq (1 - \delta)\mu] \leq \exp\left(-\frac{\delta^2 \mu}{2}\right)
\]
with $\delta = 1 - c$ gives:
\[
  p_j := \Pr[X_{j,i} \leq c\mu_j]
  \leq \exp\left(-\frac{(1 - c)^2}{2} \mu_j\right).
\]
Since $\mu_j = k_j / 6 \ \geq \log_2(\tfrac{2}{3} m_j)/3 = \ln(\tfrac{2}{3} m_j)/(3 \ln 2)$, we have:
\[
  p_j \leq \left(\tfrac{2}{3} m_j\right)^{-\alpha}.
\]

Let $Y_i$ be the indicator that bucket $i$ is bad. Then $B_j = \sum_i Y_i$ is a sum of Bernoulli variables with $\mathbb{E}[Y_i] \leq p_j$. Hence,
\[
  \mathbb{E}[B_j] \leq m_j \cdot p_j = \left(\tfrac{3}{2}\right)^\alpha m_j^{1 - \alpha}.
\]
Now apply Markov's inequality:
\[
  \Pr\left(B_j \geq \left(\tfrac{3}{2}\right)^\alpha m_j^{1 - \frac{\alpha}{2}}\right)
  \leq \frac{\mathbb{E}[B_j]}{\left(\tfrac{3}{2}\right)^\alpha m_j^{1 - \frac{\alpha}{2}}}
  \leq m_j^{-\frac{\alpha}{2}}.
\]
\end{proof}

\begin{proof}
(\cref{lem:few_total})
We first bound the total failure probability over all $L$ epochs. Since
\[
   L = \left\lceil\tfrac{c}{2}\,\log_2 M\right\rceil \leq \tfrac{c}{2}\,\log_2 M + 1,
\]
we have
\[
   m_L = M/2^L \geq \tfrac{1}{2} M^{\,1 - \frac{c}{2}}.
\]
From Lemma~\ref{lem:few}, the probability that epoch \( j \) has more than \( (\tfrac{3}{2})^\alpha m_j^{1 - \frac{\alpha}{2}} \) bad buckets is at most \( m_j^{-\alpha/2} \). Thus, the union bound gives:
\[
   \sum_{j=0}^{L-1} \Pr[\text{epoch } j \text{ fails}]
   \le \sum_{j=0}^{L-1} m_j^{-\alpha/2}
   \le L \cdot m_{L-1}^{-\alpha/2}.
\]
Since \( m_{L-1} \ge M^{1 - \frac{c}{2}} \), we have
\[
   L \cdot m_{L-1}^{-\alpha/2}
   \le \left( \tfrac{c}{2} \log_2 M + 1 \right) \cdot M^{-\beta},
   \text{where } \beta := \tfrac{\alpha}{2}(1 - \tfrac{c}{2}) > 0.
\]
As \( M \to \infty \), this product is \( o(1) \), so with high probability all epochs satisfy the per-epoch bound.

Conditioned on this event, we now sum the bad buckets over all epochs:
\[
   \sum_{j=0}^{L-1} B_j
   \le (\tfrac{3}{2})^\alpha \sum_{j=0}^{L - 1} m_j^{1 - \frac{\alpha}{2}}
   = (\tfrac{3}{2})^\alpha M^{1 - \frac{\alpha}{2}} \sum_{j=0}^{L - 1} 2^{-j(1 - \frac{\alpha}{2})}.
\]
This geometric sum is bounded by
\[
   \sum_{j=0}^\infty 2^{-j(1 - \frac{\alpha}{2})} = \frac{1}{1 - 2^{-(1 - \frac{\alpha}{2})}}.
\]
Hence the total number of bad buckets is at most
\[
   C(c)\, M^{1 - \frac{\alpha}{2}},
   \quad\text{where } C(c) = (\tfrac{3}{2})^\alpha \cdot \frac{1}{1 - 2^{-(1 - \frac{\alpha}{2})}}.
\]

\end{proof}

\section{The Frequency of Each Case in Version 3}
\label{Sec: B}
We ran Version 3 on 25 random 3-regular graphs of size 8192 and 25 random 8-regular graphs of size 4096, and computed the average percentage of each case across the different graphs. The results show that Case 1 and Case 2B occur relatively infrequently compared to Case 2A and Case 3. As expected, when the graphs become denser, Case 3 becomes dominant.

\begin{figure}[htbp]
    \centering
    \begin{minipage}[t]{0.425\textwidth}
        \centering
        \includegraphics[width=\linewidth]{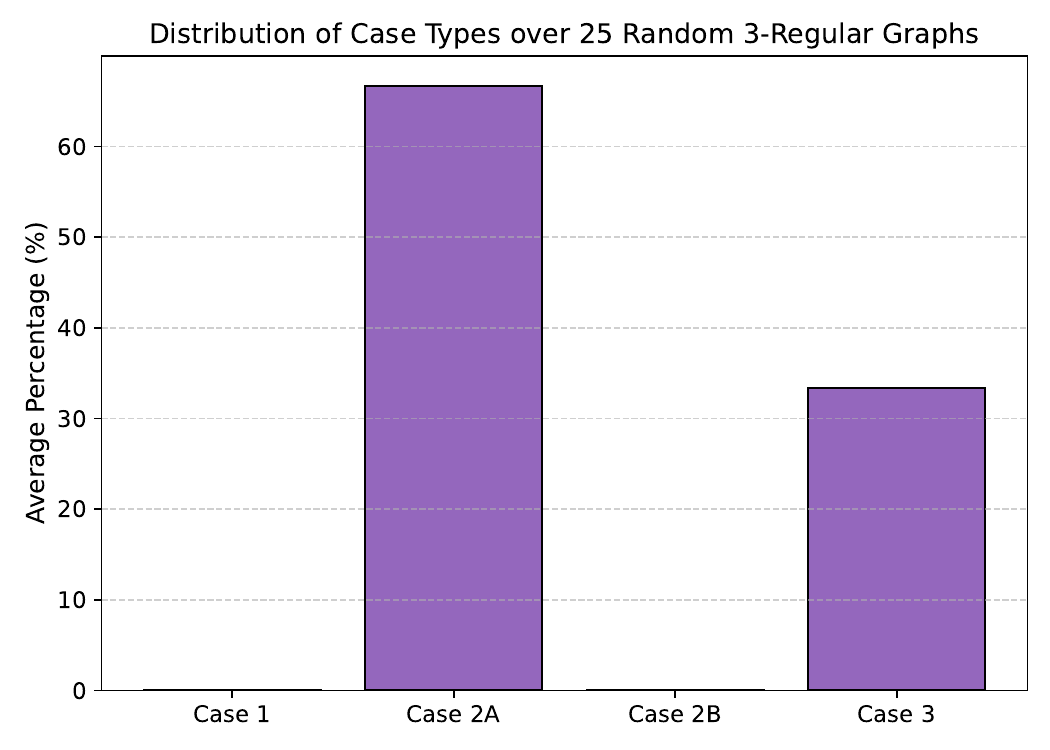}
        \caption{
        Average percentage of each case for 25 random 3-regular graphs of size 8192.
        }
        \label{fig:case_type_d3}
    \end{minipage}
    \hfill
    \begin{minipage}[t]{0.425\textwidth}
        \centering
        \includegraphics[width=\linewidth]{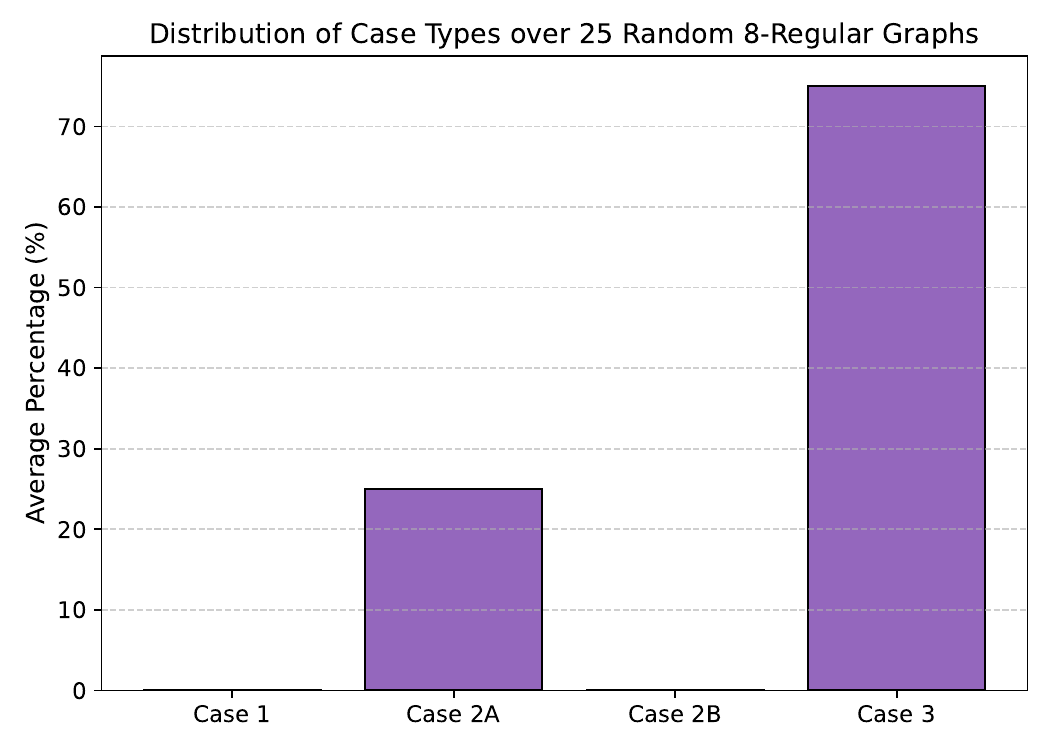}
        \caption{
        Average percentage of each case for 25 random 8-regular graphs of size 4096.
        }
        \label{fig:case_type_d8}
    \end{minipage}
\end{figure}



\newpage

\section{Process 1 results in higher remaining loads than Process 2}
\label{Sec: C}

We define three new processes, \(\mathrm{P1}\), \(\mathrm{P1a}\), and \(\mathrm{P2}\), where \(\mathrm{P1}\) is essentially equivalent to Process 1 from the main text (as we are only interested in the loads of the remaining bins), and \(\mathrm{P2}\) is a slight variant of Process 2 in which only \(k-3\) balls are thrown per round instead of \(k\). The intermediate process \(\mathrm{P1a}\) serves as a bridge to enable a direct, step-by-step comparison between \(\mathrm{P1}\) and \(\mathrm{P2}\).

\begin{theorem}[Pointwise comparison between P1 and P2]
\label{thm:P1_vs_P2}
Let \(m,k\in\mathbb{N}\) with \(m\) divisible by \(6\).
Start from any initial load vector \(L^{(0)}\in\mathbb{N}^m\).
Run the following for \(m/6\) rounds:

\begin{description}
  \item[\textbf{P1}\,:]%
    At the \emph{start} of each round, remove the three heaviest bins,
    then throw \(k-3\) balls uniformly at random into the remaining bins.
  \item[\textbf{P1a}:]%
    In each round  
    (i) mark the three heaviest bins,  
    (ii) throw \(k-3\) balls uniformly at random into \emph{all} \(m\) bins.  
    After the last round, delete all marked bins.
  \item[\textbf{P2}\,:]%
    Throw \(k-3\) balls uniformly at random into all \(m\) bins every round.  
    After the last round, delete the \(m/2\) heaviest bins.
\end{description}

Let \(L_{\mathrm{P\ast}}(R)\in\mathbb{N}^{m/2}\) be the vector of surviving bin loads
(after deletion), sorted in non-increasing order, under random realisation \(R\).
Then for every \(R\),
\[
   L_{\mathrm{P1}}(R) \ge L_{\mathrm{P1a}}(R) \ge L_{\mathrm{P2}}(R)
   \quad \text{(entrywise)}.
\]
In other words, for every outcome of the random process, the remaining bins in Process 1 are at least as heavily loaded (in every coordinate) as those in Process 2.
\end{theorem}

\begin{proof}
We couple all three processes (P1, P1a, and P2) using the same randomness. Let \(M = m\) denote the total number of bins at the beginning.

For each round \(t = 1, \dots, m/6\), we sample \(k - 3\) bins uniformly at random \emph{without replacement} from the full set \([M] = \{1, 2, \dots, M\}\), and record the selected bin labels as a list \(S_t = (b_{t,1}, \dots, b_{t,k-3})\). These selections are fixed and reused across all processes. We compare how each process uses these bin selections:

\begin{itemize}
  \item Since no bins are removed in P1a, every selected bin is valid, and each ball is placed directly into its assigned bin.

  \item In \textbf{P2}, the same placement rule applies: all bins are available during every round, and we use the same selections \(S_t\). Therefore, P1a and P2 result in the \emph{same} ball placement. The only difference is that P1a removes a fixed set of bins (those marked across rounds), while P2 removes the heaviest \(m/2\) bins at the end. Since P2 removes the worst-case set and P1a removes a possibly suboptimal one, the remaining loads in P2 can only be smaller. Hence,
  \[
    L_{\mathrm{P1a}}(R) \ge L_{\mathrm{P2}}(R).
  \]

  \item In \textbf{P1}, the top 3 heaviest bins are removed \emph{at the beginning of each round}. Let \(m_t = M - 3t\) be the number of bins remaining in round \(t\). For each selected bin label \(b_{t,j} \in S_t\), we process it as follows:
  \begin{itemize}
    \item If \(b_{t,j} \leq m_t\), then the bin exists in both P1 and P1a during round \(t\), and we place a ball into bin \(b_{t,j}\) in both processes.
    \item If \(b_{t,j} > m_t\), then that bin has been removed in P1. In this case, we resample a new bin uniformly at random from the remaining \(m_t\) bins (i.e., from \(\{1,\dots, m_t\}\)), and place the ball into the newly chosen bin in P1.
  \end{itemize}
  In contrast, P1a still places a ball into the original bin \(b_{t,j} > m_t\), which exists in P1a but has already been removed in P1.

  Therefore, every time a resampling occurs in P1, the ball is redirected to a bin that survives that round. In P1a, the same ball goes into a bin that is eventually removed. So the load of any bin that survives P1 either:
  \begin{itemize}
    \item receives the same ball as in P1a (if no resampling happened), or
    \item receives an additional ball that was redirected away from a deleted bin.
  \end{itemize}
  This shows that the load of every surviving bin in P1 is at least as large as the load of the same bin in P1a. Since both P1 and P1a delete exactly \(m/2\) bins, the remaining bins in P1 have loads that are at least as large as those in P1a:
  \[
    L_{\mathrm{P1}}(R) \ge L_{\mathrm{P1a}}(R).
  \]
\end{itemize}

\noindent
Combining the two comparisons above, we conclude:
\[
L_{\mathrm{P1}}(R) \ge L_{\mathrm{P1a}}(R) \ge L_{\mathrm{P2}}(R).
\]
\end{proof}